\documentclass[10pt, twocolumn]{article}

\usepackage{scalerel}
\usepackage{amsmath}
\usepackage{amsthm}
\usepackage{amssymb}
\usepackage[bookmarks=false]{hyperref}
\usepackage{multirow}
\usepackage{subfigure}
\usepackage{caption}
\usepackage{fullpage}

\usepackage{authblk}

\newtheorem{theorem}{Theorem}
\newtheorem{proposition}{Proposition}
\newtheorem{lemma}{Lemma}

\title{\LARGE \bf Advertising Competitions in Social Networks}
\author[1]{Antonia Maria Masucci
\thanks{Email: \href{mailto:antonia.masucci@inria.fr}{antonia.masucci@inria.fr}}}
\affil[1]{INRIA Paris-Rocquencourt\\ Domaine de Voluceau B.P. 105\\ 78153 Le Chesnay\\ France}
\author[2]{Alonso Silva
\thanks{Email: \href{mailto:alonso.silva@nokia-bell-labs.com}{alonso.silva@nokia-bell-labs.com}
To whom correspondence should be addressed.}}
\affil[2]{Nokia Bell Labs\\ Centre de Villarceaux\\ Route de Villejust\\ 91620 Nozay\\ France}
\date{}

\begin{document}
\maketitle
\thispagestyle{empty}
\pagestyle{empty}

\begin{abstract}

In the present work, we study the advertising competition of several marketing campaigns
who need to determine how many resources to allocate to potential customers
to advertise their products through direct marketing
while taking into account that competing marketing campaigns are trying to do the same.
Potential customers rank marketing campaigns according to the offers, promotions or discounts made to them.
Taking into account the intrinsic value of potential customers as well as
the peer influence that they exert over other potential customers
we consider the network value as a measure of their importance in the market
and we find an analytical expression for it.

We analyze the marketing campaigns competition from a game theory point of view, finding a closed form expression of the symmetric equilibrium offer strategy for the marketing campaigns from which no campaign has any incentive to deviate.
We also present several scenarios, such as Winner-takes-all and Borda, but not the only possible ones for which our results allow us to retrieve in a simple way the corresponding equilibrium strategy.

% 
% We are interested on the analysis of several competing marketing campaigns 
% determining to which potential customers to market 
% and how many resources to allocate to these potential customers
% while taking into account that competing marketing campaigns are trying to do the same. 
% Applying results from game theory, we characterize the equilibrium strategic resource allocation for the voter model of social networks.
% 

\end{abstract}

\section{INTRODUCTION}

In the internet age, direct marketing,
which promotes a product or service exclusively to potential customers likely to be profitable,
has brought the attention of marketing campaigns
replacing in some instances
and complementing in others 
the traditional mass marketing which promotes a product or service
indiscriminately to all potential customers.

In the context of direct marketing,
Domingos and Richardson~\cite{DomingosR2001, RichardsonD2002}
considered the {\sl network value of a customer}
by incorporating the influence of peers on the
decision making process of potential customers
deciding between different products or services promoted by
competing marketing campaigns.
If each potential customer makes a buying decision
independently of every other potential customer,
we should only consider his intrinsic value, i.e.
the expected profit from sales to him.
However, when we consider the often strong
influence potential customers exert on their peers, friends, etc.,
we have to incorporate this influence
to their network value.

Most of the existing state of the art
considers that there is
an incumbent that holds the market 
and a challenger who needs to allocate
advertisement through direct marketing
for certain individuals at a given cost of adoption
to promote the challenger product or service.
However, the cost of adoption is unknown
for most potential customers.

In the present work,
our focus is on how many resources to allocate
to potential customers, while knowing
that competing marketing campaigns are
doing the same, for them to adopt
one marketing campaign versus another.
%%% {\bf It determines the cost of adoption given
%%% the competition between marketing campaigns?}
%
We are interested
on the scenario when several competing
marketing campaigns
need to simultaneously and independently
decide how many resources to allocate
to potential customers to advertise
their products while most of the state-of-the-art focus
in only one marketing campaign (the non-simultaneous
case is also analyzed).
The process and dynamics by which
influence is spread is given
by the voter model.

\subsection{Related Works}

The general problem of influence maximization
was first introduced by Domingos and Richardson~\cite{DomingosR2001,RichardsonD2002}.
% where they studied a probabilistic setting of this problem and
% provided heuristics to compute a spread maximizing set
% which corresponds to the set of potential customers to which
% marketing campaigns need to pay for endorsements.
Based on the results of Nemhauser et al.~\cite{NemhauserWF1978},
Kempe et al.~\cite{Kempe2003,Kempe2005}
% and Mossel and Roch~\cite{MosselR2007}
provided a greedy $(1-1/e-\varepsilon)$-approximation algorithm
for the spread maximizing set.
% for very natural activation functions
% (monotone and submodular or in economic terms with decreasing marginal utility)
% the function of the expected number of active nodes
% at termination time is a submodular function
% and thus can be approximated through
% a greedy approach with a $(1-1/e-\varepsilon)$-approximation
% algorithm for the spread maximization set problem.
A slightly different model but with similar flavor,
the voter model, was introduced by Clifford and Sudbury~\cite{CliffordS1973}
and Holley and Liggett~\cite{HolleyL1975}.
In that model of social networks, Even-Dar and Shapira~\cite{EvenDar2007}
found an exact solution to the spread
maximization set.
In this work, we focus on this model of social networks
since even if the solutions are not always simple,  we can find them explicitly. 
%%%  and provided an FPTAS
%%% (Fully Polynomial Time Approximation
%%% Scheme is an algorithm
%%% that for any $\varepsilon$ approximates
%%% the optimal solution up to an error $(1+\varepsilon)$
%%% in time $\mathrm{poly}(n/\varepsilon)$)
%%% for the case in which different potential customers
%%% may have different costs.

Competitive influence in social networks has been studied
in other scenarios.
% of Bharathi et al.~\cite{BharathiKS2007},
% Sanjeev and Kearns~\cite{GoyalK2012},
% He and Kempe~\cite{HeK2013},
% Borodin et al.~\cite{BorodinFO2010}
% and Chasparis and Shamma~\cite{ChasparisS2010}.
% 
Bharathi et al.~\cite{BharathiKS2007}
proposed a generalization of the independent
cascade model of social networks %~\cite{GoldenbergLM2001}
and gave a $(1-1/e)$ approximation
algorithm for computing the best response
to an already known opponent's strategy.
% and proved that the price of competition
% of the game
% (resulting from the lack of coordination
% among the agents)
% is at most $2$.
Sanjeev and Kearns~\cite{GoyalK2012}
studied the case of two players
simultaneously choosing
some nodes to initially seed
while considering two independent functions
for the consumers
denoted switching function and
selection function.
%%% The switching function
%%% takes into account the probability of a consumer
%%% switching from non-adoption
%%% to adoption
%%% % as a function
%%% % of the fraction of his neighbors
%%% % who have adopted either of
%%% % the two products
%%% and the selection function
%%% specifies, conditional on switching,
%%% the probability
%%% that the consumer
%%% adopts one of them.
%%% % as function of the fraction
%%% % of adopting neighbors
%%% % who have adopted that option.
%%% Both players simultaneously choose
%%% some number of potential customers to initially seed.
%%% The authors
%%% make the simplifying assumption
%%% that once a potential customer is infected,
%%% it never switches again, and
%%% proceed to study some specific families of
%%% switching and adoption functions.
% We don't need to make the 
% permanent adoption assumption
% or global functions.
%%% He and Kempe~\cite{HeK2013}
%%% motivated by~\cite{GoyalK2012}
%%% studied the price of anarchy of that
%%% framework and found an upper bound
%%% of $2$ on that price.
Borodin et al.~\cite{BorodinFO2010}
showed that for a broad family of competitive
influence models is NP-hard
to achieve an approximation that
is better that the square root
of the optimal solution.
Chasparis and Shamma~\cite{ChasparisS2010}
found optimal advertising policies using dynamic programming
on some particular models of social networks.
% based on the models of~\cite{DubeyGM2006} and~\cite{Friedkin2001}.

Within the general context of competitive contests,
there is an extensive literature (see e.g.~\cite{GrossW1950,Roberson2006,masucciS2014,MasucciS2015}).
To study competitive contests, we use recent advances
of game theory, and in particular of Colonel Blotto games.
%
%%%% In the simplest version of the Colonel Blotto game,
%%%% two generals want to capture three equally valued
%%%% battlefields. Each general disposes of one
%%%% divisible unit of military resources.
%%%% The generals have to simultaneously
%%%% allocate these resources across the three battlefields.
%%%% A battlefield is captured by a general if he allocates
%%%% more resources there than his opponent.
%%%% The goal of each general is to maximize the
%%%% number of captured battlefields.
%%%% %
%%%% The relationship between Colonel Blotto games and our work is the following:
%%%% we establish a parallel between the
%%%% marketing campaigns and the generals;
%%%% and between the potential customers and
%%%% the battlefields.
%%%% Each marketing campaign needs to
%%%% strategically allocate advertising resources
%%%% to outperform competing marketing campaigns.
%%%% This needs to be done while knowing that
%%%% competing marketing campaigns are
%%%% trying to do the same.
%%%% It is thus a typical situation
%%%% in which game theory comes into play.
%%%% In our case, of course, we will not be
%%%% dealing with three potential customers
%%%% or two marketing campaigns,
%%%% so we need to extend this case to
%%%% include any number of potential customers
%%%% and any number of marketing campaigns.
%%%% By including the network value of customers,
%%%% each potential customer will not be equally
%%%% valued so we will also need to extend our work to consider
%%%% different values for different potential customers.
%
The Colonel Blotto game,
was first solved for the case of two generals and three battlefields by Borel~\cite{Borel1921,BorelV1938}.
For the case of equally valued battlefields,
also known as homogeneous battlefields case,
this result was generalized for any number of battlefields by Gross and Wagner~\cite{GrossW1950}.
Gross~\cite{Gross1950} %and Laslier~\cite{Laslier2002}
proved the existence and a method to construct the joint probability distribution.
Laslier and Picard~\cite{LaslierP2002} provided
alternative methods to construct the joint distribution by extending the method
proposed by Gross and Wagner~\cite{GrossW1950}.
Roberson~\cite{Roberson2006}
focused on the case of two generals, homogeneous battlefields and different budgets (also known as asymmetric budgets case).
Friedman~\cite{Friedman1958} studied the Nash equilibrium and best response function for
the asymmetric budgets case with two generals.
The case of two generals and where for each distinct value
there are at least three battlefields with the same value was stated and solved by
Roberson~\cite{Roberson2010} and Shwartz et al.~\cite{SchwartzLS2014}.
In the context of voting systems, Myerson~\cite{myerson1993}
found the solution for the case for equally valued battlefields with ranking scores
for any number of candidates.

The plan of this work is as follows.
In Section~\ref{sec:model} we describe the model that we are considering.
In Section~\ref{sec:results} we give the main results that we have obtained.
In Section~\ref{sec:simulations} we give simulations on some scenarios and
in Section~\ref{sec:conclusions} we conclude and describe future extensions of our work.

%%% In Section~\ref{sec:some-game-theory},
%%% we provide some preliminaries about
%%% the tools we use to find
%%% the optimal strategy.
%%% The reader acquainted with game theory
%%% and Colonel Blotto games can skip
%%% this section.
%%% In Section~\ref{sec:voter-model},
%%% we explain the voter model of social networks.
%%% In Section~\ref{sec:results},
%%% we derive the main results of this work,
%%% and in Section~\ref{sec:conclusions},
%%% we conclude this work and
%%% provide future perspectives
%%% for continuing our work.

\section{MODEL}\label{sec:model}

Consider the set of marketing campaigns~\mbox{$\mathcal{K}=\{1,2,\ldots,K\}$}
that need to allocate a certain budget, denoted by~$B$, across a set of potential
customers \mbox{$\mathcal{V}=\{1,2,\ldots,N\}$} through offers (or promotions or discounts).
% Therefore, the total budget is~$N\times B$.
Each potential customer indicates his preferences through a ranking (defined in the following subsection) of the $K$ products or services promoted
by the marketing campaigns.
For $n\in\mathcal{V}$, we denote by $w_n$ the intrinsic value of potential customer~$n$ and
by $W=\sum_{n\in\mathcal{V}} w_n$ the total intrinsic value of the set of potential customers.
Similarly, we denote by $v_n$ the network value (to be determined) of potential customer~$n\in\mathcal{V}$
and by $V=\sum_{n\in\mathcal{V}} v_n$ the total network value of the set of potential customers.

To avoid specifying the number of potential customers and dealing with the
complexities of large finite numbers, we consider the number
of potential customers to be essentially infinite.
We should, however, interpret such an infinite model as an approximation
to a large finite population with hundreds or thousands of potential customers.
We assume that campaigns' offers are independent across individual potential customers,
so that no potential customer's offers have any specific relationship with 
any other set of potential customers' offers.
%%%% That is, each potential customer's offer from each campaign is assumed to be drawn from the
%%%% campaign's offer distribution independently of the campaign's offers
%%%% to all other potential customers, and independently of all other campaigns' offers to all potential customers.
This offers' independence assumption greatly simplifies our analysis, because it allows us
to completely characterize a marketing campaign's promises by the marginal
distribution of his offers to potential customers, without saying anything more
about the joint distribution of offers to various sets of potential customers.
The infinite-population assumption (suggested and used in~\cite{myerson1993})
was introduced above essentially only to
justify this simplifying assumption of offers' independence across potential customers.

% The benefits are derived purely from the allocation
% of linearly divisible resources.

Each marketing campaign's budget constraint is expressed
as a constraint on the average offer per potential customer that a
marketing campaign can promise. Specifically,
we assume here that each marketing campaign's offer distribution for potential customer~$n$
must have mean~$B v_n/V$ to be considered credible by potential customer~$n$.
The reason is that budget $B$ should be allocated across $N$ potential customers
and each potential customer~$n$ has relative value~$v_n/V$.

%%% To avoid specifying the number of potential customers
%%% and dealing with large but finite numbers,
%%% we consider the number of potential customers to be essentially infinite.
%%% We should interpret such an infinite model as an approximation
%%% to a large finite population with thousands or millions of potential customers.
%%% {\bf 
%%% TODO: We will present an algorithm which
%%% approximates a large finite population
%%% and becomes exact when the number of potential customers is infinite.
%%% }

With a finite population of $N$ potential customers, and with a fixed budget of $B$ dollars
to be allocated, marketing campaign promises could not be
independent across all potential customers, because the offers to all potential customers would have to
sum the given budget~$B$.
However, due to Kolmogorov's strong law of large numbers, as the number of potential customers $N$ increases,
the sum of independently distributed offers with high probability will converge to the budget~$B$.
Indeed, if the mean of the campaign's offer distribution for potential customer $n\in\mathcal{V}$ is given by $Bv_n/V$
and the support of the distribution
is bounded then, for any small positive number~$\varepsilon$, $N$ potential customers' offers
that are drawn independently from the campaign's distribution would have probability
less than $\varepsilon$ of totalling more than $(1+\varepsilon)\sum_{n\in\mathcal{V}}Bv_n/V=B(1+\varepsilon)$,
when $N$ is sufficiently large.
Thus, taking the limit as the population goes to infinity,
we can assume that each campaign makes independent offers to every potential customer
and the budget constraint will hold with high probability.
%%% 
%%% {\bf We don't have de Finetti's theorem, is it a problem?}
%%% Conversely, de Finetti's theorem in probability theorem implies that,
%%% if an infinite number of potential customers are treated identically by a campaign,
%%% in the sense that their offers are exchangeable random variables,
%%% then their offers must be conditionally independent given the campaign's offer distribution.

The potential customers and their influence relationships can be modeled
as an undirected graph with self-loops~$\mathcal{G}=(\mathcal{V},\mathcal{E})$
where $\mathcal{V}$ is the set of nodes
which represent the potential customers
and $\mathcal{E}$ is the set of edges which
represent the mutual influence between potential customers.

\subsection*{Notation}
Part of the notation is summarized in Table~\ref{table:notation}.
We denote by~$\lvert\mathcal{A}\rvert$ the cardinality of set~$\mathcal{A}$.
We denote by index $k$ one of the marketing campaigns
and by index $-k$ the competing (or set of competing) marketing campaign(s) to~$k$.
For a potential customer $n\in\mathcal{V}$, we denote by $\mathcal{N}(n)$ the set of neighbors
of $n$ in graph~$\mathcal{G}$, i.e. \mbox{$\mathcal{N}(n)=\{m\in\mathcal{V}: \{n,m\}\in\mathcal{E}\}$}.
% and by $d_n$ the degree of potential customer $n$ in graph $\mathcal{G}$, i.e. $d_n=\lvert\mathcal{N}(n)\rvert$.

\begin{table*}
\caption{Notation}
\label{table:notation}
\centering
\begin{tabular}{|c|l|}
    \hline
    $\mathcal{V}=\{1,2,\ldots,N\}$ & Set of potential customers\\ \hline
    $\mathcal{K}=\{1,2,\ldots,K\}$ & Set of marketing campaigns\\ \hline
    $B$ & Total budget of marketing campaigns\\ \hline
    $w_n$ & Intrinsic value of potential customer $n$\\ \hline
    $v_n$ & Network value of potential customer $n$\\ \hline
    $W=\sum_{n\in\mathcal{V}}w_n$ & Total intrinsic value of potential customers\\ \hline
    $V=\sum_{n\in\mathcal{V}}v_n$ & Total network value of potential customers\\ \hline
    $\mathcal{G}=(\mathcal{V},\mathcal{E})$ & Graph of influence relationships\\ \hline
	% $\mathcal{N}(n)$ & Set of influence neighbors of customer $n$\\ \hline
    $M$ & Normalized transition matrix of~$\mathcal{G}$\\ \hline
    $(s_1,s_2,\ldots,s_K)$ such that  & \multirow{3}{*}{Normalized rank-scoring rule}\\
    $s_1\ge s_2\ge\ldots\ge s_K=0$, & \\
    $\sum_{j\in\mathcal{K}}s_j=1$ & \\ \hline
    $x_{k,n}$ & Offer of campaign~$k$ to customer~$n$\\ \hline
    % $\mathbf{x}_{-k,n}$ & Vector of competing campaigns' offers to $n$\\ \hline
    $\mathbf{x}_{k,\scalerel*{\cdot}{\bigodot}}$ & Vector of offers of marketing campaign~$k$\\ \hline
    $\mathbf{X}=\{x_{k,n}\}_{k\in\mathcal{K},n\in\mathcal{V}}$ & Matrix of offers\\ \hline
    $\mathbf{X}_{-k,\scalerel*{\cdot}{\bigodot}}$ & Matrix of offers of competing campaigns\\ \hline
    $\pi^{\mathrm{INT}}$ & Intrinsic payoff function\\ \hline
    $\pi$ & (Network) payoff function\\ \hline
    $u^t_n(\cdot)$ & Ranking function\\ \hline
    $f^0(\cdot)$ & Initial preferences\\ \hline
    $f^t(\cdot)$ & Preferences at time $t$\\ \hline
  \end{tabular}
\end{table*}

\subsection{Normalized rank-scoring rules}

We consider that each potential customer %$n\in\mathcal{V}$ 
ranks the set of marketing campaigns $\mathcal{K}$
in order of their offers to her.
We assume a normalized rank-scoring rule characterized by an ordered sequence
of $K$ numbers, which we denote by $s_1,s_2,\ldots,s_K$, where
\mbox{$s_1\ge s_2\ge\ldots\ge s_K=0$} and such that $\sum_{k=1}^K s_k=1$.
We consider that each potential customer~$n\in\mathcal{V}$
distributes her value~$v_n$ across marketing campaigns
according to this normalized rank-scoring rule \mbox{$\mathbf{s}=(s_1,s_2,\ldots,s_K)$} as follows:
\mbox{$v_n\mathbf{s}=(v_n s_1, v_n s_2, \ldots, v_n s_K)$}.
Thus, potential customer \mbox{$n\in\mathcal{V}$} gives the top-ranked marketing campaign
$v_n s_1$ points, the second-ranked marketing campaign $v_n s_2$, and so on,
with the $k$th ranked marketing campaign getting~$v_n s_k$ for all $k\in\mathcal{K}$.
Therefore, the payoff distributed is indeed $\sum_{k=1}^K v_n s_k= v_n\sum_{k=1}^K s_k=v_n$
where the last equality is coming from the normalization of the rank-scores.
Each marketing campaign's payoff corresponds to the sum of the payoffs
across all potential customers.

The previous assumption is not restrictive.
Given any rank-scoring rule, where $s_1,s_2,\ldots,s_K,$
are not all equal and without loss of generality $s_1\ge s_2\ge\ldots\ge s_K$,
it can be normalized to fulfill the previous statement.
Indeed, let %$S$ denote the total sum 
% of the rank-scores, i.e.
% ranking points that a customer can give, 
$S=\sum_{j=1}^K (s_j-s_K)$.
We observe that we can normalize the rank-scoring rule as follows
$(s'_1,s'_2,\ldots,s'_K)=(\frac{s_1-s_K}{S},\frac{s_2-s_K}{S},\ldots,\frac{s_K-s_K}{S})$
so that $s'_K=0$ and the sum of the rank-scores $S'$ is equal to $1$.

\subsection{Intrinsic payoff function}

We assume that the intrinsic value of potential customer $n\in\mathcal{V}$ is given by $w_n\le U$
with $U$ finite and we denote by~${\bf w}=(w_1,w_2,\ldots,w_N)$
the vector of intrinsic values of potential customers.
We consider the matrix of offers of marketing campaigns to potential customers, denoted by ${\bf X}=(x_{k,n})$,
%%%% \begin{equation*}
%%%% {\bf X}=
%%%% \left[
%%%% \begin{array}{cccc}
%%%% x_{1,1} & x_{1,2} & \ldots & x_{1,N} \\
%%%% x_{2,1} & x_{2,2} & \ldots & x_{2,N} \\
%%%% \vdots  & \vdots  & \ddots & \vdots \\
%%%% x_{K,1} & x_{K,2} & \ldots & x_{K,N}
%%%% \end{array}
%%%% \right],
%%%% \end{equation*}
where $x_{k,n}$ corresponds to the offer of marketing campaign~$k\in\mathcal{K}$ to potential customer~$n\in\mathcal{V}$.
We denote by ${\bf x}_{k,\scalerel*{\cdot}{\bigodot}}=(x_{k,1},x_{k,2},\ldots,x_{k,N})$ the vector of offers of marketing campaign~$k\in\mathcal{K}$.
We consider the matrix of offers to potential customers but only of the competing marketing campaigns to~$k$,
denoted by ${\bf X}_{{-k},\scalerel*{\cdot}{\bigodot}}$.
%%%% \begin{equation*}
%%%% {\bf X}_{{-k},\scalerel*{\cdot}{\bigodot}}=
%%%% \left[
%%%% \begin{array}{cccc}
%%%% x_{1,1} & x_{1,2} & \ldots & x_{1,N} \\
%%%% x_{2,1} & x_{2,2} & \ldots & x_{2,N} \\
%%%% \vdots  & \vdots  & \vdots & \vdots \\
%%%% x_{k-1,1} & x_{k-1,2} & \ldots & x_{k-1,N} \\
%%%% x_{k+1,1} & x_{k+1,2} & \ldots & x_{k+1,N} \\
%%%% \vdots  & \vdots  & \vdots & \vdots \\
%%%% x_{K,1} & x_{K,2} & \ldots & x_{K,N}
%%%% \end{array}
%%%% \right].
%%%% \end{equation*}
%%%% 
% and by ${\bf x}_{\scalerel*{\cdot}{\bigodot},n}= (x_{1,n}, x_{2,n}, \ldots,x_{K,n})$ the vector of offers to potential customer~$n$.
%%% 
%%% We denote by ${\bf x}_{-k,n}=(x_{1,n},\ldots,x_{k-1,n},x_{k+1,n},\ldots,x_{K,n})$ the vector of offers
%%% to potential customer~$n$ but only of the competing marketing campaigns to~$k$.

For potential customer~$n\in\mathcal{V}$, we consider a ranking function \mbox{$u_n:\mathcal{K}\rightarrow\{1,2,\ldots,K\}$}
which maps a given marketing campaign~$k$ to its ranking given by that potential customer.
For example, if marketing campaign~$k$ is the top-ranked marketing campaign
and $k'$ is the third-ranked marketing campaign for potential customer~$n\in\mathcal{V}$
then~$u_n(k)=1$ and $u_n(k')=3$.

The intrinsic payoff function for marketing campaign~$k$ is given by
\begin{equation}\label{eq:intrinsicpayoff}
\pi_k^{\mathrm{INT}}({\bf x}_{k,\scalerel*{\cdot}{\bigodot}},{\bf X}_{{-k},\scalerel*{\cdot}{\bigodot}},{\bf w})
%=\sum_{n=1}^N w_n g_n(x_{k,n},{\bf x}_{-k,n}),
=\sum_{n=1}^N w_n s_{u_n(k)},
\end{equation}
where $s_{u_n(k)}$ corresponds to the rank-score given by potential customer~$n$ for the ranking of marketing campaign~$k$.
We observe that $s_{u_n(k)}$ depends only on the offers made to potential customer~$n$.

\subsection{Evolution of the system}

We consider that time is slotted and without loss of generality we
consider that the initial time~\mbox{$t_0=0$}. 
We consider the function \mbox{$f^0:\mathcal{V}\rightarrow\mathcal{K}^K$} which maps a potential customer \mbox{$n\in\mathcal{V}$}
to her initial preferences \mbox{$f^0(n)=(f^0_1(n),f^0_2(n),\ldots,f^0_K(n))$}
where $f^0_1(n)$ corresponds to her initial top-ranked marketing campaign,
$f^0_2(n)$ corresponds to her initial second-ranked marketing campaign, and so on.
% We denote by $f^0_i(n)=1$ when potential customer $n\in \mathcal{V}$ prefers
% the product promoted by marketing campaign $i$.
% We consider that every customer has an initial preference
% between the firms, i.e.~$f^0_i=1-f^0_{-i}$.
Similarly, for $t\ge1$ we consider function \mbox{$f^t:\mathcal{V}\rightarrow\mathcal{K}^K$} which maps a potential customer \mbox{$n\in\mathcal{V}$}
to her preferences  at time~$t$, denoted by
\mbox{$f^t(n)=(f^t_1(n),f^t_2(n),\ldots,f^t_K(n))$},
where $f^t_1(n)$ corresponds to her top-ranked marketing campaign at time $t$,
$f^t_2(n)$ corresponds to her second-ranked marketing campaign at time $t$, and so on.

% Each potential customer should rank the marketing campaigns
% in order of their offers to him,
% giving $s_1 w_n$ points to the candidate who offers him
% the most, $s_2 w_n$ points to the candidate who offers him
% the second-most, and so on.
% 
% We assume that the initial preferences for a customer~$j$ is proportional to the 
% share of total advertising expenditure on customer~$j$, i.e.,
% \begin{equation}
% f_i^0(j)=
% \left\{
% \begin{array}{rl}
% 1 & \textrm{with probability } p_{i,j}(x_i,y_i)\\
% 0 & \textrm{with probability } p_{-i,j}(x_i,y_i)\\
% \end{array}
% \right.
% \end{equation}
% where the function $p_{i,j}(\cdot,\cdot)$ is given by eq.~\eqref{eq:fra1}.

The evolution of the system will be described by the voter model.
Starting from any arbitrary initial preference assignment by the potential customers
of $\mathcal{G}$, at each time $t\ge 1$, each potential customer picks uniformly at random
one of his neighbors and adopts his opinion. Equivalently,
$f^t(j)=f^{t-1}(j')$ with probability $1/\lvert\mathcal{N}(j)\rvert$ if \mbox{$j'\in\mathcal{N}(j)$}.

Similarly to the previous subsection,
for $t\ge0$ and potential customer~$n\in\mathcal{V}$, we consider function \mbox{$u_n^t:\mathcal{K}\rightarrow\{1,2,\ldots,K\}$}
which for a given marketing campaign~$k\in\mathcal{K}$ gives you the ranking of the marketing campaign for potential
customer~\mbox{$n$} at time $t$.

We are interested on the network value of a potential customer.
Following the steps of~\cite{masucciS2014}, in the next section we compute this value.

%% In other words, starting from any assignment \mbox{$f^0: \mathcal{V}\to\{0,1\}$},
%% we inductively define
%% \begin{equation}
%% f_i^{t+1}(j)=
%% \left\{
%% \begin{array}{rl}
%% 1 & \textrm{with prob. }\frac{\lvert\{j'\in N(j): f_i^t(j')=1\}\rvert}{\lvert N(j)\rvert},\\
%% 0 & \textrm{with prob. }\frac{\lvert\{j'\in N(j): f_i^t(j')=0\}\rvert}{\lvert N(j)\rvert}.\\
%% \end{array}
%% \right.
%% \end{equation}
%% 
%% For marketing campaign $k\in\mathcal{K}$ and target time $\tau$, the expected payoff is given~by
%% \begin{equation}
%% \mathbf{E}
%% \left[\sum_{n\in\mathcal{V}} s_{u_n(k)}f^\tau(n)\right].
%% \end{equation}

%% For player $i$ and target time $\tau$, the expected payoff is given~by
%% \begin{equation}
%% \mathbf{E}
%% \left[\sum_{j\in V} w_{i,j}f_i^\tau(j)\right].
%% \end{equation}
%% 

\section{RESULTS}\label{sec:results}

\subsection{Network value of a customer}

We notice that in the voter model described in the previous section, the probability
that potential customer $j$ adopts the opinion of one her neighbors $j'$
is precisely $1/\lvert \mathcal{N}(j)\rvert$. Equivalently, this is the probability
that a random walk of length $1$ that starts at $j$
ends up in~$j'$.
Generalizing this observation by induction on~$t$,
we obtain the following proposition.

\begin{proposition}[Even-Dar and Shapira~\cite{EvenDar2007}]
Let $p_{j,j'}^t$ denote the probability that a random
walk of length $t$ starting at potential customer $j$
stops at potential customer $j'$.
Then the probability that after $t$ iterations of the voter model,
potential customer $j$ will adopt the opinion that potential customer $j'$ had at time $t=0$
is precisely $p_{j,j'}^t$.
\end{proposition}

By linearity of expectation,
the expected network payoff for marketing campaign~$k\in\mathcal{K}$ at target time $\tau$,
denoted by $\pi^\tau_k$, is given by
\begin{equation*}
\pi^\tau_k=\sum_{j\in\mathcal{V}}\sum_{j'\in\mathcal{V}} w_j p^\tau_{j,j'} s_{u_{j'}^\tau(k)}.
\end{equation*}

Let $M$ be the normalized transition matrix of $\mathcal{G}$, i.e.,
$M(j,j')=1/\lvert\mathcal{N}(j)\rvert$ if $j'\in\mathcal{N}(j)$ and zero otherwise.
The probability that a random walk
of length~$\tau$ starting at~$j$ ends in $j'$
is given by the $(j,j')$-entry of the matrix~$M^\tau$.
Then
\begin{equation*}
\pi^\tau_k=\sum_{j\in\mathcal{V}}\sum_{j'\in\mathcal{V}} w_j M^\tau(j,j') s_{u_{j'}^\tau(k)}.
\end{equation*}
Therefore, the expected network payoff is given by
\begin{equation}\label{eq:networkpayoff}
\pi^\tau_k=\sum_{j'\in\mathcal{V}} v_{j'} s_{u_{j'}^\tau(k)},
\end{equation}
where the network value of potential customer $j'$ at target time~$\tau$ is given by
\begin{equation*}
v_{j'}=\sum_{j\in\mathcal{V}} w_{j}M^\tau(j,j').
\end{equation*}

We can formalize this in the following statement.
\begin{theorem}\label{theo:doumbodo}
Under the rank-scoring rule
with normalized ranking points $(s_1,s_2,\ldots,s_K)$
and intrinsic values $(w_1,w_2,\ldots,w_N)$,
the network value of potential customer $j'$ at target time~$\tau$ is given by
\begin{equation*}%\label{eq:networkvalue}
v_{j'}=\sum_{j\in\mathcal{V}} w_{j}M^\tau(j,j'),
\end{equation*}
where $M$ is the normalized transition matrix of~$\mathcal{G}$.
\end{theorem}

We notice that both eqns.~\eqref{eq:intrinsicpayoff} and~\eqref{eq:networkpayoff} are similar.
The only difference is that one considers the intrinsic value and the other
the network value (given by Theorem~\ref{theo:doumbodo}) of potential customers.
From eqns.~\eqref{eq:intrinsicpayoff} and~\eqref{eq:networkpayoff},
we obtain that after determining the network value of potential customers,
the problem of determining the resource allocation
that maximizes the expected network payoff is
similar to 
the problem of 
determining the resource allocation
that maximizes the 
expected intrinsic payoff. 
%%%% Both problems can be
%%%% reduced to determining the offers to be made by marketing campaign~$k$ to each potential customer~$n$
%%%% by taking into account this value to maximize $\sum_{n\in\mathcal{V}}v_ns_{u^\tau_n(k)}$.  
%%%% This is what we will do in the following.
Therefore, in the following we restrict ourselves to this problem.

%
% \begin{align*}
% \mathbf{P}[f_i^t(j)=1]&=\sum_{j'\in V} p^t_{j,j'}\mathbf{P}[f_i^0(j')=1]\\
% &=\sum_{j'\in V} M^t(j,j')\mathbf{P}[f_i^0(j')=1],
% \end{align*}
% 
% By linearity of expectation, we have that for player $i$
% \begin{equation}
% \mathbf{E}\left[
% \sum_{j\in V} w_{i,j}f_i^\tau(j)\right]=\sum_{j\in V}w_{i,j}\mathbf{P}[f_i^\tau(j)=1].
% \end{equation}
%%% For a subset $S\subseteq\{1,\ldots,n\}$, we denote by
%%% $1_S$ the $0/1$ column vector whose $j$-th entry is $1$ if and only
%%% if $j\in S$. 
% 
% and therefore,
% \begin{equation}\label{eq:bamako}
% \mathbf{E}\left[
% \sum_{j\in V} w_{i,j} f_i^t(j)
% \right]=
% \sum_{j\in V}\sum_{j'\in V} w_{i,j}M^t(j,j')\mathbf{P}[f_i^0(j')=1].
% \end{equation}
% We know that
% \begin{equation}
% $\mathbf{P}[f_i^0(j')=1]=p_{i,j'}(x_{i,j'},x_{-i,j'})$.
% \end{equation}
% Therefore, eq.~\eqref{eq:bamako} becomes
% \begin{equation}
% \sum_{j\in V}\sum_{j'\in V} w_{i,j}M^t(j,j')p_{i,j'}(x_{i,j'},x_{-i,j'}).
% \end{equation}
%
% Therefore, the expected payoff for player $i$ is given by
% \begin{equation}\label{eq:gabon}
% F_i({\bf x}_i,{\bf x}_{-i},{\bf v}_i)=\sum_{j'=1}^n v_{i,j'}\frac{x_{i,j'}}{x_{i,j'}+x_{-i,j'}},
% \end{equation}
% where ${\bf v}_i=(v_{i,1},v_{i,2},\ldots,v_{i,n})$.
% The previous expression is
% subject to the constraint~${\bf x}_i\in\Delta_i$.
% \begin{equation}
% \sum_{j=1}^n x_{i,j}=B_i.
% \end{equation}
% At target time~$\tau$ each potential customer~$n\in\mathcal{V}$ provides a payoff of $v_n$.

\subsection{Non-simultaneous allocations}

In this subsection, we prove that the intrinsic payoff problem is easy to solve
in the case where one marketing campaign can observe what
competing marketing campaigns are offering and after that
makes offers to potential customers.
Indeed, even in the case of two marketing campaigns,
if marketing campaign~$2$ could make offers after observing
the offers made by marketing campaign~$1$, then marketing campaign~$2$
will always be preferred by the most valuable potential customers.
For example, marketing campaign~$2$ could identify a small group
of potential customers who are the least valuable between those who are promised
strictly positive offers by
marketing campaign~$1$ (e.g. the $5\%$ of the distribution of marketing
campaign~$1$), and offer nothing to this group.
Then campaign~$2$ could offer to every other potential customer
slightly more than campaign~$1$ has promised him,
where the excess over campaign~$1$'s offers is financed
from the resources not given to the potential customers in
the first group. Every potential customer outside of the first
small group ($5\%$) would prefer marketing campaign~$2$,
who would win $95\%$ of the most valuable potential customers.
To avoid this simple outcome, we assume that
both of the marketing campaigns must make their marketing campaign
promises simultaneously.
(We may think of scenarios in which it is important to make the
first offers and in which there is a cost of delay by the response
to the first offers, but those scenarios are outside the scope of this work.)

\subsection{Family of scalable probability distributions}

We seek a solution that can be written as a family of offer
probability distributions with a scaling parameter.
We want that offers scale with the value (intrinsic or network value depending on the context) of the potential customers.
However, essentially we look for an offer distribution that has the same shape relative to this value.
We consider that the representative offer distribution~$F$ (the offer distribution with scale value~$1$)
has a probability density function~$f$ in a bounded support~$I$.
From the fundamental theorem of calculus, we have the following.
Let $I\subseteq\mathbb{R}$ be an interval and $\varphi:[a_1,b_1]\to I$
be a continuously differentiable function.
Suppose that $f:I\to\mathbb{R}$ is a continuous function.
Then
\begin{equation*}
\int_{\varphi(a_1)}^{\varphi(b_1)} f(x)\,dx = \int_{a_1}^{b_1}f(\varphi(t))\varphi'(t)\, dt. 
\end{equation*}

For potential customer~$n\in\mathcal{V}$, we use function \mbox{$\varphi(x)=x/v_n$} which is
continuosly differentiable and scales the offers by a factor~$v_n$
and therefore if the probability density function of
the representative offer density~$f(x)$ has support $[a,b]$
the scaled offer density is given by $f(x/v_n)/v_n$ and it has support $[v_na,v_nb]$.

% If a family of probability distributions
% is such that there exists a parameter $\lambda$
% for which the cumulative distribution function
% satisfies
% \begin{equation*}
% F(x;0,\lambda)= F(\lambda x;0,1),
% \end{equation*}
% then $\lambda$ is called a scale parameter.
% We consider that $F(x;0,1)$ to be the representative of this family.

We may represent marketing campaign $k$'s cumulative offer distribution
by a family of probability distributions, with representative cumulative offer distribution~$F^k(x;a,b)$,
where $F^k_n(x)=F^k(x/v_n;v_na,v_nb)$ denotes the fraction of potential customers
to whom marketing campaign $k$ will offer less than value $x$.
Each offer distribution for potential customer $n\in\mathcal{V}$
must have mean~$Bv_n/V$ and so
$F_n^k$ must be a non-decreasing function that satisfies
\begin{equation*}
\int_0^\infty x\,dF_n^k(x)=Bv_n/V,
\end{equation*}
as well as
$F_n^k(x)=0\quad\forall x\le0$,
and
\begin{equation*}
\lim_{x\to+\infty}F_n^k(x)=1.
\end{equation*}

\subsection{Symmetric equilibrium}

A symmetric equilibrium of the marketing campaign competition is a scenario
in which every marketing campaign is expected to use the
same offer distribution, and each marketing campaign
finds that using this offer distribution maximizes its
chances of winning when the
other marketing campaigns are also simultaneously and
independently allocating their offers according
to this distribution
(and all potential customers perceive that the $K$ marketing campaigns
have the same probability of winning the market).
In this work, we focus exclusively on finding
such symmetric equilibria.

In the following, we prove that
there is a symmetric equilibrium which corresponds to
a family of probability distributions
with scale parameter $v_n$ for 
potential customer $n\in\mathcal{V}$.
Let $F(x)=F(x;a,b)$ denote the representative cumulative distribution function acting as
the equilibrium strategy and
let $F_n(x):=F(x/v_n;v_na,v_nb)$ denote the cumulative distribution function
representing the equilibrium offer distribution for potential customer~$n\in\mathcal{V}$.
$F_n(x)$ denotes the cumulative probability that a given potential customer~$n$
will be offered less than~$x$ by any other given marketing campaign,
according to this equilibrium distribution.
% We notice that we have dropped the dependency on $k$
% since the equilibrium is used by every marketing campaign.

Consider the situation faced by a given marketing campaign~$k$
when it chooses its offer distribution, assuming that every other
marketing campaign will use the equilibrium offer distribution.
When marketing campaign~$k$ offers $x$ to potential customer~$n$,
the probability that this marketing campaign $k$ will be ranked in position $j$
by potential customer~$n$ is given by $P(j,F_n(x))$ where we let
\begin{equation*}
P(j,q)={K-1\choose j-1}	q^{K-j}(1-q)^{j-1}.
\end{equation*}
That is, $P(j,q)$ denotes the probability that exactly $j-1$
of the $K-1$ competing marketing campaigns will offer more than $x$,
given that each other marketing campaign
has an independent probability $q$ of offering less than $x$
to this potential customer.
Equivalently, $P(j,q)$ denotes the probability that exactly $K-j$
of the $K-1$ competing marketing campaigns
will offer less than $x$.

If marketing campaign $k$ offers $x$ to potential customer~$n$,
then the expected value that this potential customer will give to this marketing campaign is $R_n(F_n(x))$
where
\begin{equation*}
R_n(q)=v_n\sum_{j=1}^K P(j,q)s_j.
\end{equation*}

%%%% If marketing campaign $k$ fixes an offer $x$ and for all $n\in\mathcal{V}$ it offers $v_nx$
%%%% to potential customer~$n$, then the expected value that the marketing campaign will obtain
%%%% is given by $R(G(x))$ where
%%%% \begin{equation*}
%%%% R(q)=\frac 1N\sum_{n\in\mathcal{V}} v_n\sum_{j=1}^K P(j,q)s_j.
%%%% \end{equation*}

Things could be more difficult if there were a positive probability
of other marketing campaigns offering exactly~$x$, but we can ignore
such complications because 
we will prove (see Lemma~\ref{lemma:sudan}) that the equilibrium distribution
cannot assign positive probability to any single point.
When all marketing campaigns independently use the same offer distribution,
they must all get the same expected score from potential customer~$n$
which must equal~$v_n/K$.

%%% There is a symmetric equilibrium in which all candidates
%%% use the representative cumulative distribution~$G$ if and only if,
%%% for any other distribution~$F$ that is on the nonnegative numbers and has mean~$B/V$,
%%% \begin{equation}
%%% \int_0^\infty R(F(x))\,dG(x)\le\frac V K.
%%% \end{equation}
%%% 
% For each marketing campaign~$k$, we consider the one dimensional marginal distribution functions
% $\{G_k^n\}$ for each customer~$n\in\mathcal{V}$.

% Given the strategies of the other players, each marketing campaign~$k$ maximizes the sum of the expected
% payoffs across potentional customers.
% \begin{equation}
% \max\sum_{n=1}^N\int_0^\infty v^n F_{-k}^n(x)\,dF_k^{n},
% \end{equation}

\begin{theorem}\label{theo:tembine}
In a $K$-marketing campaign competition
under the normalized rank-scoring rule $(s_1,s_2,\ldots,s_K)$
and values $(v_1,v_2,\ldots,v_N)$,
there is a unique scalable symmetric equilibrium
of the marketing campaigns' offer-distribution game.
In this equilibrium, each marketing campaign
chooses to generate offers according
to a family of probability distributions, with scale parameter
$v_n$ for potential customer $n$, that has support on the
interval from $0$ to $s_1KBv_n/V$,
and which has a cumulative distribution $F(\cdot)$
that satisfies the equation
\begin{equation*}
x=R_n(F_n(x))/(V/KB),\quad\forall x\in[0,s_1KBv_n/V].
\end{equation*}
\end{theorem}

% \begin{proof}

The proof follows the steps of Theorem~$2$ in~\cite{myerson1993}.
The following is a constructive proof and we decompose the proof in the next following lemmas.
% \end{proof}

\begin{lemma}\label{lemma:sudan}
If there is a symmetric equilibrium distribution of offers, it must be continuous,
i.e. it cannot have any points of positive probability.
\end{lemma}

\begin{proof}
If all marketing campaigns used a representative offer distribution $F(\cdot)$
that assigned a positive probability $\delta$ to some point $x>0$,
then there would be a positive fraction $\delta^K$ of potential customers
who would be exactly indifferent among the marketing campaigns
since they receive from each of them the same offer.
Any marketing campaign could then increase his average point score among this group
by giving an arbitrarily small increase (say, $\varepsilon$)
to most of the potential customers to whom he was going to offer $x$
and the cost of this increase could be financed by moving
an arbitrarily small fraction of this group down to zero.
In other words, if the offer distribution had a positive mass
at some point, then a marketing campaign could gain
a positive group of potential customers by a transfer
of resources that would lower his score from
only an arbitrarily small number of potential customers.
\end{proof}

%%%% We will first prove that
%%%% if there is a symmetric equilibrium distribution of offers,
%%%% it must be continuous, i.e. it cannot have any points of positive probability.
%%%% Indeed, if all marketing campaigns used a representative offer distribution $F(\cdot)$
%%%% that assigned a positive probability $\delta$ to some point $x>0$,
%%%% then there would be a positive fraction $\delta^K$ of potential customers
%%%% who would be exactly indifferent among the marketing campaigns
%%%% since they receive from each of them the same offer.
%%%% Any marketing campaign could then increase his average point score among this group
%%%% by giving an arbitrarily small increase (say, $\varepsilon$)
%%%% to most of the potential customers to whom he was going to offer $x$
%%%% and the cost of this increase could be financed by moving
%%%% an arbitrarily small fraction of this group down to zero.
%%%% In other words, if the offer distribution had a positive mass
%%%% at some point, then a marketing campaign could gain
%%%% a positive group of potential customers by a transfer
%%%% of resources that would lower his score from
%%%% only an arbitrarily small number of potential customers.

\begin{lemma}
We have that
\begin{equation*}
R_n(0)=s_Kv_n=0,\quad R_n(1)=s_1v_n,
\end{equation*}
and $R_n(\cdot)$ is a continuous and strictly increasing function over the interval
from $0$ to $1$.
\end{lemma}

\begin{proof}
These equations hold because $P(j,0)$ equals $0$ unless $j$ equals $K$,
$P(j,1)$ equals $0$ unless $j$ equals $1$, and $P(K,0)=1=P(1,1)$.
Continuity of $R_n(\cdot)$ follows directly from the formulas,
because $R_n(q)$ is polynomial in~$q$.
Let us show that $R_n(\cdot)$ is increasing.
First, we verify that
\begin{equation*}
R_n(q)=v_n\sum_{j=2}^K (s_{j-1}-s_j)\sum_{m<j} P(m,q),
\end{equation*}
using $s_K=0$.
We observe that
$\sum_{m<j}P(m,q)$
denotes the probability that more than $K-j$ other marketing campaigns
have made offers in an interval of probability $q$,
and this probability must be a strictly increasing
function of~$q$.
The ordering of the $s_j$ values guarantees that at least one term
in this $R_n(q)$ expression must have a positive $(s_{j-1}-s_j)$
coefficient, and none can be negative. Therefore $R_n(\cdot)$ is an increasing function.
\end{proof}

%%%% Applying the definitions of $P(\cdot,\cdot)$ and $R_n(\cdot)$,
%%%% we now show that
%%%% \begin{equation}
%%%% R_n(0)=s_Kv_n=0,\quad R_n(1)=s_1v_n,
%%%% \end{equation}
%%%% and $R_n(\cdot)$ is a continuous and strictly increasing function over the interval
%%%% from $0$ to $1$.
%%%% %
%%%% These equations hold because $P(j,0)$ equals $0$ unless $j$ equals $K$,
%%%% $P(j,1)$ equals $0$ unless $j$ equals $1$, and $P(K,0)=1=P(1,1)$.
%%%% %
%%%% Continuity of $R_n(\cdot)$ follows directly from the formulas,
%%%% because $R_n(q)$ is polynomial in~$q$.
%%%% %
%%%% Let us show that $R_n(\cdot)$ is increasing.
%%%% First, we verify that
%%%% %
%%%% \begin{equation*}
%%%% R_n(q)=v_n\sum_{j=2}^K (s_{j-1}-s_j)\sum_{m<j} P(m,q),
%%%% \end{equation*}
%%%% using $s_K=0$.
%%%% %
%%%% We observe that
%%%% \begin{equation*}
%%%% \sum_{m<j}P(m,q)
%%%% \end{equation*}
%%%% denotes the probability that more than $K-j$ other marketing campaigns
%%%% have made offers in an interval of probability $q$,
%%%% and this probability must be a strictly increasing
%%%% function of~$q$.
%%%% The ordering of the $s_j$ values guarantees that at least one term
%%%% in this $R_n(q)$ expression must have a positive $(s_{j-1}-s_j)$
%%%% coefficient, and none can be negative. Therefore $R_n(\cdot)$ is an increasing function.

\begin{lemma}
The lowest permissible offer~$0$ must be in the support
of the equilibrium distribution of offers.
\end{lemma}

\begin{proof}
The main idea is that, if the minimum of the support were strictly
greater than zero, then a marketing campaign would
be devoting positive resources to potential customers near
the minimum of the support of the distribution.
He would expect to get almost no value ($s_K=0$)
from these potential customers, because all other marketing campaigns would almost surely
be promising them more.
Thus, it would be better to reduce the offers to $0$
for most of these potential customers in order to make serious offers
for at least some of them.

The above argument can be formalized as follows.
Because, as we have shown before, there are no points of positive probability,
the cumulative offer distribution $F_n(\cdot)$ for potential customer~$n$ is continuous.
Let $z$ denote the minimum of the support of the equilibrium offer distribution for potential customer~$n$,
so $F_n(z)=0$ but $F_n(z+\varepsilon)>0$ for all positive~$\varepsilon$.
% Therefore for potential customer~$n$, we have that~$F_n(v_nz)=0$ but $F_n(v_n(z+\varepsilon))>0$.
%
Now, select any fixed $y$ such that $y>z$ and $F_n(y)>0$.
For any $\varepsilon$ such that $0<\varepsilon<y-z$,
a marketing campaign might consider deviating from the equilibrium
by promising either $y$ or $0$ to each potential customer~$n$
in the group of potential customers whom he was supposed to offer between $z$ and $(z+\varepsilon)$,
according to his $F_n$-distributed random-offer generator.
The potential customers in this group were going to be given
offers that averaged some amount between $z$ and $(z+\varepsilon)$,
so he can offer $y$ dollars to at least a $z/y$ fraction of these potential customers
without changing his offers to any other potential customer.
Among this $z/y$ fraction of the group,
he would get an average point score of $R_n(F_n(y))$,
by outbidding the other marketing campaigns who are using
the $F_n$ distribution;
so the deviation would get him an average point score
of at least $(z/y)R_n(F_n(y))$ from this group of potential customers
(the potential customers moved down to zero in this deviation
would give him $s_Kv_n=0$ points).
If he follows the equilibrium,
however, he gets at most $R_n(F_n(z+\varepsilon))$
as his average point score from this group of potential customers.
So to deter such a deviation, we must have
$(z/y)R_n(F_n(y))\le R_n(F_n(z+\varepsilon))$,
and so 
\begin{equation*}
z\le y\frac{R_n(F_n(z+\varepsilon))}{R_n(F_n(y))}.
\end{equation*}
But $R_n(F_n(z+\varepsilon))$ goes to $R_n(F_n(z))=R_n(0)=0$
as $\varepsilon$ goes to $0$, and so $z$ must equal $0$.
\end{proof}

\begin{lemma}
There is some positive constant $\alpha$ such that
\begin{equation*}
R_n(F_n(x))=\alpha x.
\end{equation*}
\end{lemma}

\begin{proof}
Let $x$ and $y$ be any two numbers in the support
of the equilibrium distribution for potential customer~$n$ such that \mbox{$0<x<y$}.
A marketing campaign could deviate by taking a group
of potential customers to whom he is supposed to give offers
close to $x$, according to his equilibrium plan,
and instead he could give them offers close to $y$
to an $x/y$ fraction of this group and he could offer $0$
to the remaining $(1-x/y)$ fraction.
Because the support of the representative distribution
contains $0$ as well as $x$ and $y$,
neither this self-financing deviation
nor its reverse (offering close to $x$ to
a group of potential customers of whom an $x/y$ fraction were
supported get close to $y$,
and the remaining $(1-x/y)$ fraction
were supposed to get close to~$0$)
should increase the marketing campaign's expected average point
score from this group of potential customers.
Thus, we must have
\begin{equation*}
R_n(F_n(x))=(x/y)R_n(F_n(y))+(1-x/y)R_n(F_n(0)).
\end{equation*}
But $R_n(F_n(0))=R_n(0)=0$,
so we obtain \[ \frac{R_n(F_n(x))}{x}=\frac{R_n(F_n(y))}{y}, \]
for all $x$ and $y$
in the support of the equilibrium offer distribution for potential customer~$n$.
So there is some positive constant $\alpha$ such that,
for all $x$ in the support of the offer distribution for potential customer~$n$,
$R_n(F_n(x))=\alpha x$.
\end{proof}

\begin{lemma}
We have that the constant $\alpha=V/KB$.
\end{lemma}
\begin{proof}
The mean offer must equal $Bv_n/V$ under the $F_n$ distribution, therefore
\begin{equation*}
\int_0^{s_1v_n/\alpha} x\,dF_n(x)=B\frac{v_n}{V}.
\end{equation*}

We also know that a marketing campaign who uses the same offer distribution $F_n$ as all the other marketing campaigns must expect the average point score $v_n/K$, so
\begin{align*}
\frac{v_n} K&=\int_0^{s_1v_n/\alpha} R_n(F_n(x))\,dF_n(x)=\int_0^{s_1v_n/\alpha} \alpha x\, dF_n(x)\\&=\alpha B\frac{v_n}{V}.
\end{align*}
\end{proof}

From the previous lemma, the support of the $F_n$ distribution is the interval from~$0$ to $s_1v_n/\alpha=s_1KBv_n/V$,
and the cumulative distribution satisfies the formula
\begin{equation*}
R_n(F_n(x))=\frac{V}{KB}x,\quad\forall x\in[0,s_1KBv_n/V].
\end{equation*}

%%%% To evaluate the constant $\alpha$, we use the fact that the mean offer
%%%% must equal $Bv_n/V$ under the $F_n$ distribution, so
%%%% \begin{equation*}
%%%% \int_0^{s_1v_n/\alpha} x\,dF_n(x)=B\frac{v_n}{V}.
%%%% \end{equation*}
%%%% 
%%%% We also know that a marketing campaign who uses the same offer distribution $F_n$ as all the other marketing campaigns must expect the average point score $v_n/K$, so
%%%% \begin{align*}
%%%% \frac{v_n} K&=\int_0^{s_1v_n/\alpha} R_n(F_n(x))\,dF_n(x)=\int_0^{s_1v_n/\alpha} \alpha x\, dF_n(x)\\&=\alpha B\frac{v_n}{V}.
%%%% \end{align*}
%%%% 
%%%% Thus, the support of the $F_n$ distribution is the interval from~$0$ to $s_1v_n/\alpha=s_1KBv_n/V$,
%%%% and the cumulative distribution satisfies the formula
%%%% \begin{equation*}
%%%% R_n(F_n(x))=\frac{V}{KB}x,\quad\forall x\in[0,s_1KBv_n/V].
%%%% \end{equation*}
%%%% 

\begin{lemma}\label{lemma:doumbodo}
$F_n$ is an equilibrium.
\end{lemma}

\begin{proof}
In general, for any nonnegative~$x$, we have $R_n(F_n(x))\le\frac{V}{KB}x$,
because when $x>s_1KBv_n/V$ $R_n(F_n(x))=R_n(1)=s_1v_n<\frac{V}{KB}x$.
So using any other distribution $G_n$, that has mean $Bv_n/V$
for potential customer~$n\in\mathcal{V}$ and is on the nonnegative numbers,
would give to a marketing campaign an expected score
\begin{align*}
\int_0^\infty &R_n(F_n(x))\,dG_n(x)\le\int_0^\infty\frac{V}{KB} x\,dG_n(x)\\
&=\frac V {KB}\int_0^\infty x\,dG_n(x)=\frac{v_n}{K},
\end{align*}
with equality if the support of $G_n$ is contained in the interval $[0,s_1KBv_n/V]$.
Thus, no marketing campaign can increase his expected score by deviating from $F_n$ to
some other distribution, when all other marketing campaigns are using the distribution~$F_n$.
\end{proof}

%%%% It only remains to verify that we do get an equilibrium when $F_n$ satisfies this formula.
%%%% In general, for any nonnegative~$x$, we have $R_n(F_n(x))\le\frac{V}{KB}x$,
%%%% because when $x>s_1KBv_n/V$ $R_n(F_n(x))=R_n(1)=s_1v_n<\frac{V}{KB}x$.
%%%% %
%%%% So using any other distribution $G_n$, that has mean $Bv_n/V$
%%%% for potential customer~$n\in\mathcal{V}$ and is on the nonnegative numbers,
%%%% would give to a marketing campaign an expected score
%%%% \begin{align*}
%%%% \int_0^\infty R_n(F_n(x))\,dG_n(x)&\le\int_0^\infty\frac{V}{KB} x\,dG_n(x)\\
%%%% &=\frac V {KB}\int_0^\infty x\,dG_n(x)\\
%%%% &=\frac{v_n}{K},
%%%% \end{align*}
%%%% with equality if the support of $G_n$ is contained in the interval $[0,s_1KBv_n/V]$.
%%%% Thus, no marketing campaign can increase his expected score by deviating from $F_n$ to
%%%% some other distribution, when all other marketing campaigns are using the distribution~$F_n$.
%%%% 
%%%% \end{proof}
Lemmas~\ref{lemma:sudan}-\ref{lemma:doumbodo} are the proof of Theorem~\ref{theo:tembine}.
The previous theorem provides us a method to obtain explicitly
the cumulative offer distribution functions under different ranking-scores.

\section{SIMULATIONS}\label{sec:simulations}

\subsection*{Winner-takes-all}

We notice that our problem is more general than
a simple pairwise competition between marketing campaigns.
For the pairwise competition there already exists a solution (see e.g.~\cite{SchwartzLS2014}).
However, 
a pairwise competition is not always what is needed.
For example, consider the case when each customer chooses
only one marketing campaign to buy a product
(it could be for example buying a house, in which most
of the potential customers will buy only one house).
To see this, consider the example of three
competing marketing campaigns $X$, $Y$, and $Z$
and five equally valuable customers (for the sake of simplification).
Consider the pure strategies
\begin{align*}
{\bf x}&=(0.2,0.2,0.2,0.2,0.2),\\
{\bf y}&=(0.0,0.0,0.0,0.5,0.5),\\
{\bf z}&=(0.5,0.5,0.0,0.0,0.0).
\end{align*}
In that case, the pairwise competition gives
that marketing campaign $X$ captures $3$ out of $5$ potential customers to $Y$
(the first three),
and that $X$ captures $3$ out of $5$ potential customers to $Z$ (the last three),
thus winning in a pairwise competition against both marketing competitors.
However, since each customer will only choose
one product, the final outcome will be
$2$ customers for $Y$, $2$ customers for $Z$,
while only $1$ customer for $X$.

The case where the objective is to be the first evaluated marketing campaign
and being second does not provide any value can be represented as follows:
\begin{equation*}
s_1=1,\quad s_2=0,\quad\ldots,\quad s_K=0.
\end{equation*}

In that case $R_n(q)=v_n P(1,q)=v_n q^{K-1}$.
Therefore from Theorem~\ref{theo:tembine} the equilibrium cumulative distribution satisfies
\begin{equation*}
x=(F(x))^{K-1}KBv_n/V,\quad\forall x\in[0,KBv_n/V],
\end{equation*}
and thus
\begin{equation*}
F(x)=\left(\frac{x}{KBv_n/V}\right)^{1/(K-1)},\quad\forall x\in[0,KBv_n/V].
\end{equation*}

When $K=2$ we recover the result of~\cite{SchwartzLS2014} for pairwise competition.
It is also interesting to notice the similarity between this solution
and the characterization of the solution for an all-pay auction with one object~\cite{BayeKV1996}.
We notice that there is a tight relationship between this scenario, Colonel Blotto games and auctions.
A Colonel Blotto game can be seen as a simultaneous all-pay auction
of multiple items of complete information.
An all-pay auction is an auction in which every bidder
must forfeit its bid regardless of whether it wins the object
which is awarded to the highest bidder.
It is an auction of complete information
since the value of the object is known to every bidder.
In other contexts, this was already noted by Szentes and Rosenthal~\cite{SzentesR2003},
Roberson~\cite{Roberson2006} and Kvasov~\cite{Kvasov2007}.
%
% We have $k$ players (or bidders) denoted~$X^1,X^2,\ldots,X^k$ and $n$ objects.
% A player $X^i$ has budget $B_i$ for $1\le i\le K$
% and it can bid to an object~$j$ a proportion of his budget~$x^i_j$ for $1\le j\le n$
% and $1\le i\le K$.
% We assume that every player has the same total divisible budget,
% i.e., $B_i=B$ and without loss of generality we assume this budget to be equal to $1$.
%
% The payoff to player $X$ is then
% \begin{equation*}
% \sum_{i=1}^n a_i1_{\{x_i=\max\limits_{1\le j\le k} x_j\}}.
% \end{equation*}
%
% Thus, player~$X$ will need to maximize its expected payoff.

% \begin{figure*}[htp]
% \centering
% \subfigure[Winner-takes-all equilibrium offer distribution when we consider a budget of $1000$ dollars, $K=2$ and~\mbox{$v_n/V=1/20$}; $K=4$ and \mbox{$v_n/V=1/40$};
% and $K=6$ and \mbox{$v_n/V=1/60$} (for them to have the same support).]{\includegraphics[scale=0.49]{plot3.eps}}\quad
% \subfigure[Borda equilibrium offer distribution when we consider a budget of $1000$ dollars, $K=2$, $K=4$ and $K=6$ and relative value~\mbox{$v_n/V=1/20$}. We notice that the equilibrium offer distribution of Borda is independent of the number of competing marketing campaigns.]{\includegraphics[scale=0.49]{plot2.eps}}
% \caption{Equilibrium offer distributions under Winner-takes-all and Borda ranking scoring rules. }\label{fig:pato1}
% \end{figure*}

\begin{figure}[htp]
\centering
\includegraphics[scale=0.49]{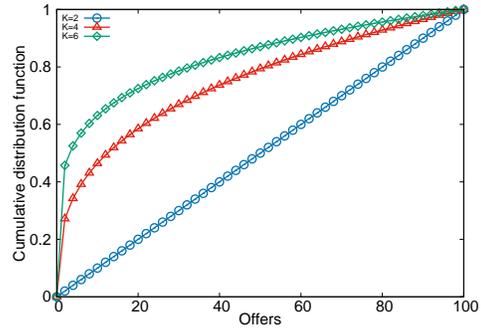}
\caption{Winner-takes-all equilibrium offer distribution when we consider a budget of $1000$ dollars, $K=2$ and~\mbox{$v_n/V=1/20$}; $K=4$ and \mbox{$v_n/V=1/40$};
and $K=6$ and \mbox{$v_n/V=1/60$} (for them to have the same support).}
\end{figure}

\begin{figure}[htp]
\centering
\includegraphics[width=0.5\textwidth]{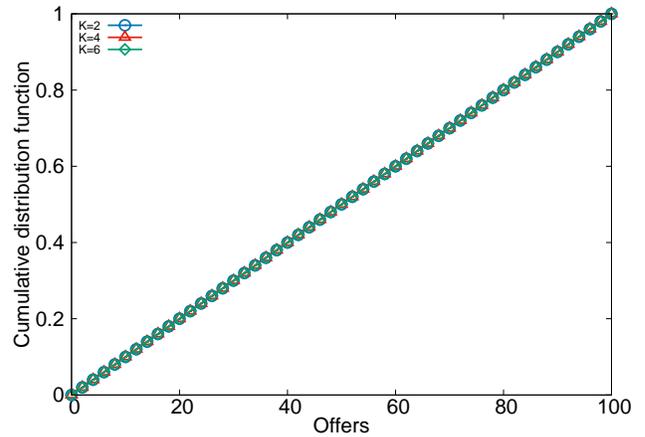}
\caption{Borda equilibrium offer distribution when we consider a budget of $1000$ dollars, $K=2$, $K=4$ and $K=6$ and relative value~\mbox{$v_n/V=1/20$}. 
We notice that the equilibrium offer distribution of Borda is independent of the number of competing marketing campaigns.}
\label{fig:pato1}
\end{figure}

% \begin{figure}[h!]
%   \centering
%     \includegraphics[width=0.5\textwidth]{profits.eps}
%   \caption{Profits for the incumbent (defender)
%         vs the difference of valuations between communities ($\delta$) for equal budgets.}\label{fig:rwanda2}
% \end{figure}

Figure~\ref{fig:pato1}(a) gives us the equilibrium offer distribution
when we consider that the budget of each marketing campaign is $1000$~dollars for
three different competing scenarios:
\begin{itemize}
\item there are $2$ marketing campaigns and the relative value of a customer is $v_n/V=1/20$;
\item there are $4$ marketing campaigns and the relative value of a customer is $v_n/V=1/40$;
\item there are $6$ marketing campaigns and the relative value of a customer is $v_n/V=1/60$.
\end{itemize}
The chosen parameters in the three scenarios allow us to consider the same support of the offers distributions.

We observe that when there are two competing marketing campaigns, the equilibrium offers are made uniformly at random
over the support interval from $0$ to $100$ dollars.
However, increasing the number of competing marketing campaigns, we observe
that marketing campaigns offers are skewed offering less than the average
to most of the potential customers while offering much more than the average
for a reduced number of potential customers.
In particular, for four marketing campaigns,
more than $50\%$ of the potential customers receive offers of less than $14$ dollars (the average offer is $25$ dollars).
This effect is even more pronounced for six marketing campaigns where more than
$50\%$ of potential customers receive offers of less than $4$ dollars (the average offer is $17$ dollars).

% \begin{figure*}[htp]
%   \centering
% \begin{subfigure}
% \includegraphics[scale=0.7]{./Figures/plot3.eps}
% \caption{Winner-takes-all equilibrium offer distribution.}
% \label{fig:pato1}
% \end{subfigure}
% \begin{subfigure}
% \includegraphics[scale=0.7]{./Figures/plot2.eps}
% \caption{Borda equilibrium offer distribution.}
% \label{fig:pato2}
% \end{subfigure}
% \caption{Equilibrium offer distributions under Winner-takes-all and Borda ranking scoring rule when we consider $K=2$ and~\mbox{$v_n/V=1/20$}; $K=4$ and \mbox{$v_n/V=1/40$};
% and $K=6$ and \mbox{$v_n/V=1/60$} (for them to have the same support). We notice that the equilibrium offer distribution of Borda is independent of the number of competing marketing campaigns.}
% \end{figure*}

\subsection*{Borda}

Another interesting case is when the rank-scoring rule is 
linearly decreasing with the ranking (we denote it Borda
for its similarity to Borda ranking votes). 
For example, it can be given by
\begin{equation*}
% s_1=(K-1)/S,s_2=(K-2)/S,s_3=(K-3)/S,\ldots, s_K=0,
s_1=\frac{(K-1)}S,s_2=\frac{(K-2)}S,s_3=\frac{(K-3)}S,\ldots, s_K=0,
\end{equation*}
where $S=\sum_{j=1}^K s_j=K(K-1)/2$.

The function~$R_n(q)$ under that rule is given by
\begin{align*}
R_n(q)&=v_n\sum_{j=1}^K P(j,q)\frac{2(K-j)}{K(K-1)}\\
% &=\frac{2v_n}{K}\sum_{j=1}^K P(j,q)\frac{K-j}{K-1}\\
&=\frac{2v_n}K\sum_{j=1}^K\binom{K-1}{j-1}q^{K-j}(1-q)^{j-1}\frac{K-j}{K-1}\\
&=\frac{2v_n}K\sum_{j=0}^{K'}\binom{K'}{j}q^{K'-j}(1-q)^j\left(1-\frac{j}{K'}\right)\\
&=\frac{2v_n}K\left(1-\frac{K'(1-q)}{K'}\right)=\frac{2v_n}Kq
\end{align*}
where we have made the change of variable $K'=K-1$ and use
the formula of the expected value of a binomial distribution.

% The function~$R_n(F(x))$ under that rule corresponds to the
% expected value of~$2(K-j)/K(K-1)$,
% where $j$ is a random variable such that $K-j$ has a binomial distribution
% with expected value $K(K-1)F(x)/2$.
Thus, by the previous theorem
\begin{align*}
\frac{V}{KB}x% &=R_n(F(x))=v_n\sum_{j=1}^K P(j,F(x))\frac{2(K-j)}{K(K-1)}\\
&=\frac{2v_n}{K} F(x), \quad\forall x\in[0,2Bv_n/V].
\end{align*}
Therefore,
\begin{equation*}
F(x)=\frac{x}{2Bv_n/V}\quad\forall x\in[0,2Bv_n/V].
\end{equation*}
Therefore, the equilibrium offer distribution under this rule
is a uniform distribution over the interval from $0$ to $2Bv_n/V$.
We notice that the equilibrium offer distribution is independent
of the number of competing marketing campaigns~$K$.

Figure~\ref{fig:pato1}(b) gives us the equilibrium offer distribution
when we consider that the budget for each marketing campaign is $1000$ dollars,
the relative value of a customer is $v_n/V=1/20$
and we consider three scenarios with $K=2$, $K=4$, and $K=6$.

We observe that in these three scenarios the equilibrium offer distribution is uniformly distributed
over the suppport interval from $0$ to $100$ dollars and it is {\sl independent on the number
of competing marketing campaigns}.

The previously considered scenarios, Winner-takes-all and Borda, are two out of many possible scenarios that can be analyzed and to which our previous results
can be applied.

% \begin{figure}[h!]
% \centering
% \includegraphics[width=0.48\textwidth]{./Figures/plot3.eps}
% \caption{Equilibrium offer distribution under winner-takes-all ranking-scoring rule.
% There are $K=4$ marketing campaigns, $B=1$ thousand dollars and \mbox{$v_n/V=1/40$}.}
% \label{fig:pato1}
% \end{figure}
% 
% \begin{figure}[h!]
% \centering
% \includegraphics[width=0.48\textwidth]{./Figures/plot2.eps}
% \caption{
% There are $K=6$ marketing campaigns, $B=1$ thousand dollars and \mbox{$v_n/V=1/20$}.}
% \label{fig:pato2}
% \end{figure}

% There are $K=4$ marketing campaigns, $B=1$ thousand dollars and \mbox{$v_n/V=1/40$}.

\section{Conclusions}\label{sec:conclusions}

In this work, we studied advertising competitions in social networks. In particular, we analyzed the scenario of several marketing campaigns determining to which potential customers to market and how many resources to allocate to these potential customers while taking into account that competing marketing campaigns are trying to do the same.

As a consequence of social network dynamics, the importance of every potential customer in the market can be expressed in terms of her network value which is a measure of the influence exerted among her peers and friends and of which we provided an analytical expression for the voter model of social networks.

Defining rank-scoring rules for potential customers and
using tools from game theory, we have given a closed form expression of the symmetric equilibrium offer strategy for the marketing campaigns from which no campaign has any interesting to deviate.
Moreover, we presented some interesting out of many possible scenarios to which our results can be applied.

\bibliographystyle{hieeetr}
\bibliography{mybibfile}
\end{document}